\documentclass{svproc}
\usepackage{url}

\usepackage{cite}
\usepackage{amsfonts}
\usepackage{amsmath}
\usepackage{tikz}
\usepackage{pgfplots}
\pgfplotsset{compat=1.14} 
\usepgfplotslibrary{fillbetween}
\usepackage{pgfplotstable}

\begin{document}
\mainmatter              
\title{Individual Fairness for Social Media Influencers}
\titlerunning{Individual Fairness for SM Influencers}  
%

\author{Stefania Ionescu \and Nicolò Pagan \and Anikó Hannák}
\authorrunning{Stefania Ionescu et al.} 
%
\tocauthor{Stefania Ionescu, Nicolò Pagan, Anikó Hannák}
\institute{Social Computing Group, University of Zürich, Zürich 8050, Switzerland,\\
\email{ionescu@ifi.uzh.ch}, \email{pagan@ifi.uzh.ch}, \email{hannak@ifi.uzh.ch},
}

\maketitle              

\begin{abstract}
Nowadays, many social media platforms are centered around content creators (CC). On these platforms, the tie formation process depends on two factors: (a) the exposure of users to CCs (decided by, e.g., a recommender system), and (b) the following decision-making process of users. Recent research studies underlined the importance of content quality by showing that under exploratory recommendation strategies, the network eventually converges to a state where the higher the quality of the CC, the higher their expected number of followers. 
In this paper, we extend prior work by (a) looking beyond averages to assess the fairness of the process and (b) investigating the importance of exploratory recommendations for achieving fair outcomes. Using an analytical approach, we show that non-exploratory recommendations converge fast but usually lead to unfair outcomes. Moreover, even with exploration, we are only guaranteed fair outcomes for the highest (and lowest) quality CCs.
\let\thefootnote\relax\footnotetext{This is a preprint of the following chapter: Stefania Ionescu, Nicolò Pagan, and Anikó Hannák , `Individual Fairness for Social Media Influencers', published in Complex Networks and Their Applications XI, Volume I, edited by Hocine Cherifi, Rosario Nunzio Mantegna, Luis M. Rocha, Chantal Cherifi, Salvatore Micciche, 2023, Springer reproduced with permission of Springer Nature Switzerland AG 2023. The final authenticated version is available online at: \url{https://doi.org/10.1007/978-3-031-21127-0_14}.}
\end{abstract}
\section{Introduction}



The past couple of decades brought a steep increase in the impact social media platforms have on our lives, e.g., by shaping the information we receive \cite{bakshy2012role} and the opinions we form \cite{hall2018brexit}. 
During this time, platforms previously designed to only connect real-life friends slowly encouraged users to follow strangers based on their content.
Today, platforms such as YouTube, Twitter, Instagram, and TikTok are heavily centred around User Generated Content (UGC) and use recommender systems (RSs) to facilitate the exploration of content.
In response to this change, some users specialize in creating even semi-professional content that can attract more and more followers to the point that they can make revenue based on their audience.
Similarly to the labour market \cite{adams1963towards}, it is then natural to expect that these online platforms guarantee fairness for the content creators (CCs) in a way that equally good CCs should be rewarded similarly in terms of visibility and audience, and ultimately of income.
Since the network formation process is often heavily mediated by the RS, it is thus appropriate to ask whether they produce fair outcomes for the CCs.



To first understand the structure and the properties of these social media platforms, \cite{pagan2021meritocratic} proposed a simple model in which (a) each CC has an intrinsic and objective quality, (b) in each round, users receive a recommendation for a CC (which could be drawn from the uniform distribution, or from a preferential attachment (PA) process), and (c) users follow the recommended CC if this CC has a higher quality than all the user's current followees. By simulating this model, Pagan et al. \cite{pagan2021meritocratic} showed that the expected number of followers (in-degree) of the CCs follow a Zipf's law \cite{zipf2016human}. In particular, the expected rankings of CCs given by their quality and by the number of followers are the same.

While these results suggest that such a network-formation process leads to a fair outcome for the CCs, we observe two important limitations. First, the analysis is restricted to two exploratory recommendation processes (i.e., processes which take risky actions in order to uncover better options). However, in practice, not all RSs are as such. Recent literature in RSs argued why and how we should encourage diversity by finding the right balance between exploration and exploitation \cite{mcnee2006being, kunaver2017diversity, helberger2018exposure, gravino2019towards}. 
In practice, even pairs of popular items might not be jointly accessible to users, i.e., if a user is recommended and follows one of the two, they will not be recommended the other
\cite{guo2021stereotyping}. 
This puts real-world RS in stark contrast with PA and uniform random (UR) recommendations where every user can be recommended any CC.
Second, the authors only focus on the \emph{expected} number of followers \emph{at convergence} (i.e., after all users were eventually recommended the best CC). However, such ex-ante fairness does not imply ex-post fairness (i.e., even if in expectation CCs receive a number of followers proportional to their quality, many of the actual outcomes that could materialize are unfair) \cite{myerson1981utilitarianism}. 
Moreover, as also noted by the authors \cite{pagan2021meritocratic}, there might be long times to convergence. This means that even if a fair outcome would eventually be reached, this might not happen within reasonable time.

This paper aims to address these two challenges by bridging between the network science, RS, and algorithmic fairness communities. More precisely, the current work (a) defines both ex-ante and ex-post fairness metrics for CCs, (b) extends the model with extreme PA (a non-exploratory RS, inspired by the popularity RS\cite{chaney2018algorithmic}, which only recommends the most followed CCs), (c) uses Markov Chains to theoretically study the network formation process and its fairness under extreme PA, and (d) compares the network formation processes and their fairness under extreme PA with the ones under PA and UR recommendations (either by referring to prior work, or by novel analytical and numerical results).

\section{Related Work}
\paragraph{Networks.} 
Over the years, the complex networks community developed a variety of 
simple yet realistic mechanisms that explain the formation of social networks (e.g., the small-world network model \cite{watts1998collective}, and the preferential attachment model -PA- \cite{barabasi1999emergence}). In PA, newborns form connections to existing nodes with a probability proportional to the degree. 
This rich-get-richer phenomenon successfully reproduces the idea that popular users experience higher visibility, which in turns brings them more popularity. 
On the other hand, it gives little emphasis on the socio-economic microscopic foundations that explain why individuals make certain connections. 
Focusing on this alternative approach, one line of research in sociology (Stochastic Actor Oriented Models \cite{snijders1996stochastic}) and one in economics (strategic network formation models \cite{jackson2010social}) take an utilitarian perspective: agents build their connections to maximize some benefit e.g., their network centrality. 
The quality-based model of Pagan et al. \cite{pagan2021meritocratic} combines these approaches by using a UR or PA-based RSs and a utilitarian following decision-making function for users.
To enhance our understanding on the coupling between RS and human network behavior, we add a non-exploratory RS and investigate the fairness of the resulting outcomes.

\paragraph{Fairness.} Researchers are not only concerned with the average performances of processes, but also with the equity of these processes in impacting individuals. This is reflected by the extensive work on developing fairness measures as well as a methodology to choose the most suitable one depending on the application domain \cite{verma2018fairness, garg2020fairness, mitchell2021algorithmic}. From this various fairness metrics, we focus on \emph{individual fairness} which assesses the degree similarly qualified individuals receive similarly quality outcomes (see \cite{binns2020apparent} for an overview of its importance and apparent incompatibility with other fairness metrics). An important phenomenon is the \emph{timing effect}, according to which it is not enough to specify a welfare function but also when this should be measured (ex-ante or ex-post)\cite{myerson1981utilitarianism}. Building on these, we define and investigate both the ex-ante and ex-post individual fairness for CCs.

\paragraph{Recommender Systems.} Recently, the RS community argued for the importance of looking beyond accuracy \cite{mcnee2006being} in RS-evaluation. Design-wise, there is an ongoing struggle to develop diverse RSs \cite{kunaver2017diversity, helberger2018exposure, gravino2019towards} which, moreover, ensure that any two items could be recommended jointly to users \cite{guo2021stereotyping}. This perhaps explains why the popularity-based algorithms implemented within the RS community \cite{chaney2018algorithmic, lucherini2021t} differ from PA \cite{barabasi1999emergence, pagan2021meritocratic} by not allowing for a complete exploration of alternatives. Our motivation to analyse extreme PA is grounded in this prior work,
although our goal is to understand the network-formation process and its fairness.


\section{Model}
As mentioned in the previous two sections, our work is based on the quality-based model proposed by \cite{pagan2021meritocratic} and briefly presented below.
As any model, this is based on some simplifying assumptions, for which we refer the user back to the original paper. 
Note that our notation occasionally differs from \cite{pagan2021meritocratic} in order to allow for future extensions. The next two subsections contain our advancements of the model: (a) a RS-based formulation of the network formation process (plus the definition of the new RS, extreme PA), and (b) the fairness definitions.

We assume users can be partitioned into $n\geq 2$ content creators (CCs) and $m$ (non-CC) users, with $m \gg n$. We refer to CCs as $CC_1, CC_2, \dots, CC_n$ and assume they are ordered strictly by an existing objective quality. That is, $CC_1$ is the absolute best CC, $CC_2$ the next best CC, and so on. Every user can follow any content creator thus leading to a bipartite, unweighted, and directed network. We represent this network by its adjacency matrix $A\in \{0,1\}^{m\times n}$, such that (s.t.) $a_{u, i}$ is $1$ if user $u$ follows $CC_i$, and $0$ otherwise. 

The network formation is a sequential dynamic process where (a) the network is initially empty, and (b) at each timestep, each user is recommended a CC which they can follow or not. Formally, let $A^t$ capture the state of the network at time $t$. By assumption (a), $a^0_{u, i} = 0$ for each user $u\in \overline{m}$ and CC $i\in \overline{n}$. 
\footnote{Thought the paper we use $\overline{k}$ to denote the set of non-zero natural numbers that are at most equal to $k$, i.e. $\overline{k} := \{1, 2, \dots k\}$.}
The next subsection presents how recommendations are made and what are their effects.

\subsection{Recommendation Function}

We view recommendations as functions which map users to items, i.e., the recommendations at time $t$ are given by $R^t: \overline{m} \to \overline{n}$.
\footnote{Note that under full generality each user receives a \textit{list} of recommendations. However, allowing for a single recommendation per user makes our model comparable with  \cite{pagan2021meritocratic}.} The recommender system is the algorithm which produces such a recommender function based on the current state of the network. In this paper, we look at three such algorithms which we list below from the most to the least exploratory one:
\begin{itemize}
    \item \emph{Uniform Random (UR).} Each user can be recommended any CC with an equal probability, i.e., $R_{\text{UR}^t(u)} \sim \mathcal{U}(\overline{n})$.
    
    \item \emph{Preferential Attachment (PA).} Each content creator can be recommended with a chance proportional to their current number of followers. That is, if  $a_{., i} := \sum_{u\in \overline{m}} a_{u, i}$ is the number of followers of $CC_i$, then:
    $$\mathbb{P}(R^t_{\text{PA}}(u) = i) = \frac{a^t_{., i} +1}{ \sum_{j\in \overline{n}} (a^t_{., j} +1)}.$$ 
    
    \item \emph{ExtremePA.} An extreme version of preferential attachment which recommends only (one of) the most popular items:
    \[ 
    \mathbb{P}(R^t_{\text{ext}}(u) = i) = \left\{
    \begin{array}{ll}
          0 & \text{, if } a^t_{., i} < \max_j a^t_{., j},\\
          1/|\{i\text{, s.t. } i \in \arg \max_j a^t_{., j}\} | & \text{, otherwise.}\\
    \end{array} 
    \right. 
    \]
    
\end{itemize}

After each user $u$ receives their associated recommendation $R^t(u)$ they decide to follow $R^t(u)$ iff $R^t(u)$ is better (perhaps vacuously) than any followee of $u$:
\[ 
    a^{t+1}_{u, i} = \left\{
    \begin{array}{ll}
          1 & \text{, if } R^t(u) = i \text{ and } i<j\ \forall j\ \text{ s.t. } a^t_{u, j} = 1
          \\
          a^{t}_{u, i} & \text{, otherwise.}\\
    \end{array} 
    \right. 
    \]

\subsection{Metrics of Interest}
\label{sec:metrics}
Since the network at time $t+1$ only depends on the network at time $t$, we can view this network formation process as a Markov Chain (MC). More details will follow in Section~\ref{sec:theoretical_results}. With this interpretation, we first want to investigate (a) whether this is an absorbing MC (i.e., whether there exist some networks - absorbing states - which will eventually be reached and which will not change no matter how many timesteps proceed), and, if so, (b) the expected number of steps until absorption (i.e., until reaching such a state). 

However, as mentioned in the introduction, we are also interested in \emph{the fairness} of the overall process, both ex-ante and ex-post. Thus we define:
\begin{itemize}
    \item \emph{Ex-post individual fairness for CCs.} We say that a network $A$ is \textit{(individually) fair} if the number of followers of CCs have the same ranking as the CCs, i.e., if $a_{., 1} \geq a_{., 2} \geq \dots \geq a_{., n}$. We also say that $A$ is \textit{(individually) $CC_i$-fair} if $CC_i$ is one of the top $i$ ranked CCs, i.e. if $|\{j: a_{., j} > a_{., i}\}| < i$. In particular, this means that a network is fair iff it is $CC_i$-fair for all $i\in \overline{n}$.
    \footnote{Note that this is a weak version of the fairness definition. Alternatively, we can say a network $A$ is fair if $a_{., 1} > a_{., 2} > \dots > a_{., n}$. However, as we will see later in the results section, as the number of users goes to infinity the chance of achieving equality goes to zero. So, in the limit, the two definitions are equivalent.}
    
    \item \emph{Ex-ante individual fairness for CCs.} When this is an absorbing MC, we can also look at the ex-ante fairness of the network formation process. More precisely, we say that a process is \textit{ex-ante (individually) fair} if the expected number of followers of CCs at absorption is decreasing, i.e., if $\mathbb{E}[a^{\infty}_{., 1}] \geq \mathbb{E}[a^{\infty}_{., 2}] \geq \dots \geq \mathbb{E}[a^{\infty}_{., n}]$. Similarly, we say that a process is \textit{ex-ante (individually) $CC_i$-fair} if $|\{j: \mathbb{E}[a^{\infty}_{., i}] < \mathbb{E}[a^{\infty}_{., j}]\}| < i$.

\end{itemize}


\section{Results}
\label{sec:theoretical_results}
In our results section we investigate the metrics of interest introduced in Section~\ref{sec:metrics}. As such, subsection~\ref{sec:abs_MC} shows how this process can be viewed as a MC and what are the absorbing states. Subsection~\ref{sec:time_abs} looks at the expected time to absorption. Finally, in subsection~\ref{sec:fairness}, we build on the prior two and investigate fairness. Each of these subsections starts with (a) a summary paragraph which gives an overview of the results, and (b) a take-away paragraph which discusses the relevance of these results. The reminder of each subsection is used for the results themselves. In the interest of space, in a few of the later proofs we omit details and only provide the proof's outline and intuition. 

\subsection{An Absorbing Markov Chain}
\label{sec:abs_MC}
\emph{Summary.} This subsection starts by proving that our process is an absorbing MC (see Theorem~\ref{thm:MC}). During the proof we also define the transition matrix. Fig.~\ref{fig:MC_small} illustrates this process for the small case of two CCs and two (non-CC) users. The remaining results characterize the absorbing states. First, Theorem~\ref{thm:absorbing_PA} shows that under exploratory recommendation processes (i.e., PA or UR) the absorbing states are the ones where each user follows the best CC. Second, Theorem \ref{thm:absorbing_extreme} shows that, under ExtremePA, a state reachable from the $\textbf{0}$ state is absorbing if (a) there is a unique CC with the maximum number number of followers and (b) no new user would like to follow this CC. For the latter, we use Lemma~\ref{lem:ExtremePA_states}, which states that any state reachable from \textbf{0} has, for each $CC_i$ with the maximum number of followers, a user who does not follow any CC better than $CC_i$.

\emph{Take-away.} The results of this subsection build a representation for the process which facilitates its understanding. Importantly, by describing the transition matrix we see the differences between the three RSs: (a) ExtremePA has more absorbing states, (b) out of these, the states where $CC_1$ is the most followed CC (i.e., the $CC_1$-fair ones) are the absorbing states under the exploratory RSs, (c) PA situates itself between ExtremePA and UR in terms of exploration. By the latter we mean that, when compared to UR, PA could remain for longer times in transient states (which are fairness-wise similar to the states which are absorbing only under ExtremePA). 



\begin{figure}%
    \centering
    \includegraphics[width=\textwidth]{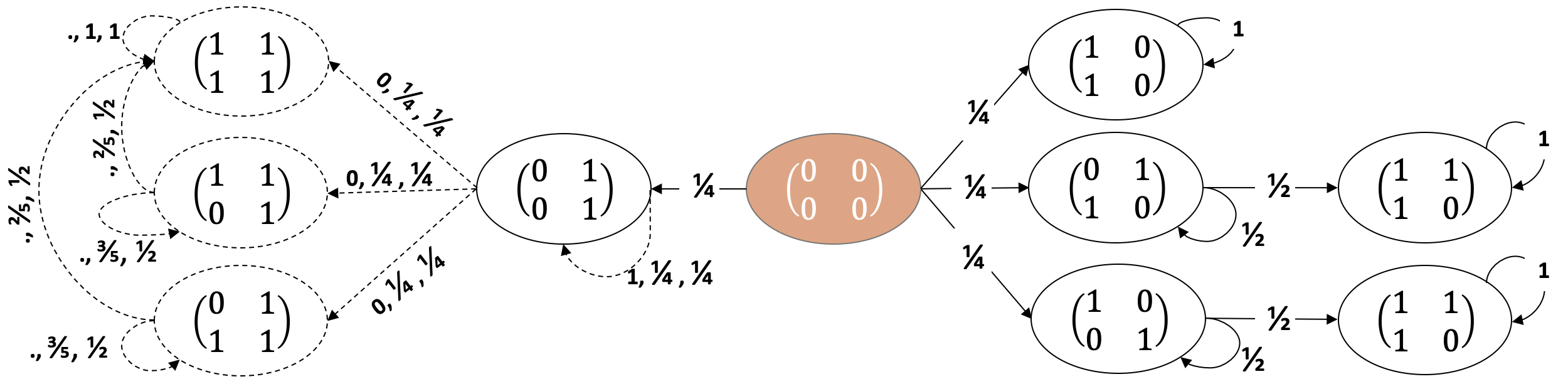}
    \caption{An example of the MC representation for $n = m = 2$. The coloured node is the starting state. We use (a) full edges for transitions that are the same under the three RSs, and (b) dotted edges for transitions that differ (labeled with the respective probabilities for (1) ExtremePA, (2) PA, (3) UR). Dots replace probabilities when the starting state is not reachable from \textbf{0} under ExtremePA.
    }
    \label{fig:MC_small}
\end{figure}

\begin{theorem}
\label{thm:MC}
$(A^t)_{t\geq 0}$ is an absorbing MC (for all three RSs).
\end{theorem}
\begin{proof}
First, $(A^t)_{t\geq 0}$ is a MC, as $A^{t+1}$ depends only on the network configuration at the previous step. More precisely, it is a MC where: (a) the state space is $\{0, 1\}^{m\times n}$, (b) the initial distribution is $\lambda$ where $\lambda_{A} = 1$ iff $A$ is the zero matrix, and (c) the transition matrix given by $p_{B, C} := \mathbb{P}(A^{t+1} = C|A^t = B)$, where:
\begin{itemize}
        \item from the zero matrix we can only transit to a matrix where each user follows exactly one CC, i.e.:
        $
            p_{\textbf{0}, C} = \left\{
            \begin{array}{ll}
                  1/{n^m} & \text{, if } c_{u, .} = 1 \ \forall \ u \in \overline{m} \\
                  0 & \text{, otherwise}\\
            \end{array} 
            \right. 
            $;
        
        \item from any other matrix we can only transit to a new matrix where each user either (a) follows the same CCs as before, or (b) follows exactly one more CC which is better than the best CC they followed so far; the probabilities of such transitions depend on the recommendation process. 
\end{itemize}

Moreover, from the shape of the transition matrix it follows that this process is an absorbing MC. To see this, note that $p_{B, C}$ is non-zero iff $C=B$ or there is some user $u$ who follows one more CC, i.e. $c_{u, .} > b_{u, .}$. So, any state $B$ is either absorbing, or can transit to a state of a strictly higher sum of elements. Since the sum of elements of $B$ is bounded by $m \cdot n$, such a sequence of transitions must be finite. Thus, in this latter case, an absorbing state will be eventually reached through a sequence of transitions. \hfill$\qed$
\end{proof}


\begin{theorem}
\label{thm:absorbing_PA}
Under PA and UR recommendations, a state $B$ is absorbing iff all users follow the best CC, i.e. iff $b_{u, 1} = 1$ for all $u\in \overline{m}$.
\end{theorem}
\begin{proof}
($\Rightarrow$) 
Assume there exist some user $u\in \overline{m}$ s.t. $b_{u, 1} = 0$. Under PA and UR recommendations there is always a non-zero chance $u$ is recommended $CC_1$. If this happens, then $u$ follows $CC_{1}$. So we can transit to a state $C\neq B$ where $c_{u, 1} = 1$ with a non-zero probability. Hence, $B$ is not absorbing. 

($\Leftarrow$) The reverse is straightforward. If $b_{u, 1} = 1$, then no recommendation will change the followees of $u$. This hods for all users $u$, so 
$B$ is absorbing.
\hfill$\qed$
\end{proof}


\begin{lemma}
\label{lem:ExtremePA_states}
For any state $B$ reachable from \textbf{0} and any $i\in \overline{n}$ s.t. $b_{., i} = \max_j b_{., j}$ there exits some user $u \in \overline{m}$ s.t. $b_{u, j} = 0$ for all $j<i$.
\end{lemma}
\begin{proof}
We use an inductive argument. The claim is obviously true for any state directly reachable from \textbf{0} (i.e., any network achievable in one timestep), as, in such states, each user follows exactly one CC. Next, we assume the claim is true for some state $B$. From this state, we either remain in $B$, or we transit to a state $C$ where a subset of these CCs have a higher maximum number of followers. As each $CC_i$ in this subset increased their number of followers they needed to be recommended to a user $u$ that preferred them to any CC they previously followed. As such, $CC_i$ is the best CC followed by $u$ in $C$, i.e., $i = \min \{j: c_{u, j} = 1\}$.
\hfill$\qed$
\end{proof}

\begin{theorem}
\label{thm:absorbing_extreme}
Under Extreme PA,
all absorbing states $B$ reachable from \textbf{0} are such that (a) there exist some $i \in \overline{n}$ s.t. $b_{., i}> b_{., j}$ for all $j \in \overline{n}-\{i\}$, and (b) for all $u\in \overline{m}$ there exist some $j\leq i$ s.t. $b_{u, j} = 1$. Moreover, every state which satisfies (a) and (b) is absorbing.
\end{theorem}
\begin{proof}
For the first part, 
let $B$ be an absorbing state reachable from \textbf{0}.
If (a) is not the case then there are at least two CCs, say $CC_i$ and $CC_j$ with $i<j$, that have the maximum number of followers. By Lemma~\ref{lem:ExtremePA_states}, there exists some user $u$ who does not follow $CC_i$ yet. Therefore, if everybody is recommended $CC_i$, only the number of followers of $CC_i$ will increase by at least one (as $u$ will follow them). There is a non-zero probability of this happening, and, in such a case, a new state is reached. Therefore, $B$ is not an absorbing state if (a) does not hold. Next, if (a) holds but (b) does not then in the next round (when the unique most followed $CC_i$ is recommended to everybody) new users will follow $CC_i$. So, again, if (a) holds but (b) does not then $B$ cannot be absorbing. 

The second part is straightforward. If (a) holds then only the unique CC with the highest number of followers can be recommended in the next round. Since (b) holds, no new user will follow this CC. So, there is no chance of transiting to a new state, i.e. the current state is absorbing. \hfill$\qed$
\end{proof}


\subsection{Expected Time to Absorption}
\label{sec:time_abs}
\emph{Summary.} Theorem~\ref{thm:time_abs} shows that as the number of users goes to infinity, the expected time to absorption under ExtremePA is less than $2$. Two preliminary results are needed to prove this statement: (a) the chance of having ties in the number of followers after the first round of recommendations (Lemma~\ref{lem:eq->0}) and (b) the chance that in the second round no one new would follow the most followed $CC_i$ (if $i\neq n$, see Lemma~\ref{lem:prob_converged}) both go to $0$ as the number of users goes to infinity
\footnote{Note that the two preliminary results are independent on the RS.}.
These results are used for the annotations in Fig.~\ref{fig:MC_extreme}, which summarizes the process and provides the intuition for the proof of the theorem. 

\emph{Take-away.} This puts ExtremePA in sharp contrast with the exploratory RSs, as prior work indicates that, under the UR and PA recommendation scenarios, the convergence time increases logarithmically in the number of CCs and linearly (or sub-linearly) in the number of users
(see Figs.~2 and 7b of \cite{pagan2021meritocratic}). Importantly, it indicates that, while  for ExtremePA it could be sufficient to analyse the fairness in the absorbing states alone, this might not be the case for PA and UR: as these RSs could lead, in practice, to long convergence times the fairness in the transient states should also be evaluated. This is particularly true for PA, which, as discussed earlier, could remain in unfair states for longer than UR.



\begin{lemma}
\label{lem:eq->0}
For any $i\neq j \in \overline{n}$, $\mathbb{P}(a^1_{., i} = a^1_{., j}) \to 0$ as $m\to \infty$.
\end{lemma}
\begin{proof}
First, let $Y_u$ be the following random variables depending on the recommendations in the first round:
\[ 
    Y_u = \left\{
    \begin{array}{ll}
          0 & \text{ , if user } u \text{ is recommended some } CC_k \text{ with } k \neq i, j; \\
          1 & \text{ , if user } u \text{ is recommended } CC_i;\\
          -1 & \text{ , if user } u \text{ is recommended } CC_j.
    \end{array} 
    \right. 
    \]
Since the first recommendations are uniform random, $(Y_u)_{u\in \overline{m}}$ are i.i.d. (with $\mathbb{P}(Y_u = 0) = (n-2)/n$ and $\mathbb{P}(Y_u = c) = 1/n$  when $c = \pm 1$). Hence, $\mathbb{E}[Y_u] = 0$ and $\text{Var}(Y_u)  = \mathbb{E}[Y_u^2] = 2/n$. By the central limit theorem it follows that,
$$\frac{\sum_{u\in \overline{m}} Y_u - m\cdot 0}{2/n \cdot \sqrt{m}} \to^{\mathcal{D}} \mathcal{N}(0, 1).$$

In particular, this implies that $\lim_{m\to \infty} \mathbb{P}(a^1_{., i} = a^1_{., j}) = 0$.
\hfill$\qed$
\end{proof}

\begin{lemma}
\label{lem:prob_converged}
For any $i \in \overline{n-1}$, $\mathbb{P}( (a^1_{., i} > a^1_{.,j}\ \forall j\in \overline{n}) \wedge (\forall u \in \overline{m},\  \exists j \in \overline{i} \mbox{ s.t.}\ a^1_{u, j} = 1)) \to 0$ as $m\to \infty$.
\end{lemma}
\begin{proof} This follows as the latter part of the conjunction corresponds to the scenario of all users being recommended one of the top $i$ CCs in the first round. So, $\mathbb{P}( (a^1_{., i} > a^1_{.,j}\ \forall j\in \overline{n}) \wedge (\forall u \in \overline{m},\  \exists j \in \overline{i}\ a^1_{u, j} = 1)) \leq \mathbb{P}(\forall u \in \overline{m},\  \exists j \in \overline{i}\ a^1_{u, j} = 1)= \left( \frac{i}{n} \right)^m \to 0$ as $m\to \infty$. \hfill$\qed$
\end{proof}

\begin{figure}%
    \centering
    \includegraphics[width=\textwidth]{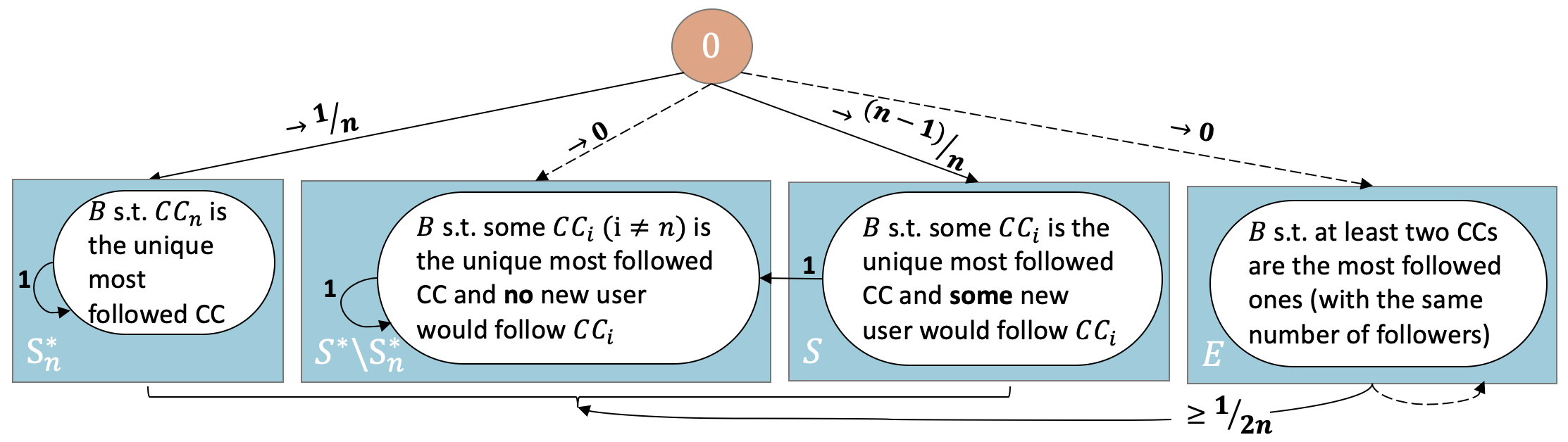}
    \caption{A summary of the MC for ExtremePA as $m\to \infty$. Rectangles represent sets of states. Transition probabilities are annotated in the limit. Dotted lines correspond to transitions who are negligible as $m\to \infty$.}
    \label{fig:MC_extreme}
\end{figure}

\begin{theorem}
\label{thm:time_abs}
Under ExtremePA, the expected time to absorption goes to $2-1/n$ as $m$ goes to infinity.
\end{theorem}
\begin{proof}
Based on Theorems~\ref{thm:MC} and \ref{thm:absorbing_extreme}, we can group states in subsets, as shown in Fig.~\ref{fig:MC_extreme}. Let $\mu_B$ be the expected time from the state $B$ to absorption. Then $\mu_{B: B\in S^*} = 0$ (as all states in $S^*$ are absorbing) and $\mu_{B: B\in S} = 1$ (as all states in $S$ lead in one timestep to a state in $S^*$). Therefore,
$$\mu_{\textbf{0}} = 1\cdot p_{\textbf{0}, S_n^*} + 1\cdot p_{\textbf{0}, S-S_n^*} + 2\cdot p_{\textbf{0}, S} + \sum_{B\in E} (1 + \mu_B) \cdot p_{\textbf{0}, B} 
$$

We can use Lemmas~\ref{lem:eq->0} and \ref{lem:prob_converged} to find the transition probabilities between the sets of states as $m\to \infty$: (a) $p_{\textbf{0}, S^*_n}\to \frac{1}{n}$, (b) $p_{\textbf{0}, S^*-S^*_n}\to 0$, (c) $p_{\textbf{0}, S}\to \frac{n-1}{n}$, and (d) $p_{\textbf{0}, E} \to 0$. Since, in addition, $\sum_{B\in E} p_{\textbf{0}, B} \cdot (1 + \mu_B) \leq p_{\textbf{0}, E}\cdot (1+ \mu_{B^*})$ (where $B^* = \arg \max_{B\in E} \mu_B$) and $\mu_{B^*}\leq c$ ($c$ constant)
\footnote{This can be easily shown for $c= 4n/(n-1)$ by using Lemma~\ref{lem:ExtremePA_states} to prove that $p_{B^*, S\cup S^*}\geq 1/2n$ (although better bounds can be obtained).}
, the result follows. 
\hfill$\qed$
\end{proof}

\subsection{Fairness for Content Creators}
\label{sec:fairness}
\emph{Summary.} Building on previous results, we show that the exploratory RSs are both ex-ante and ex-post $CC_1$-fair (Corollary~\ref{cor:CC1_fair_PA}), while ExtremePA is only ex-ante fair (Theorem~\ref{thm:ex_ante_ExtremePA}). In fact, Corollary~\ref{cor:CC1_fair_ExtremePA} shows the probability of achieving a $CC_1$-fair absorbing state under ExtremePA goes to $1/n$ as the number of users goes to infinity, while Corollary~\ref{cor:CC_fair_ExtremePA} shows that the probability of achieving a fair outcome for all CCs goes to $1/n!$. Finally, Fig.~\ref{fig.fairness} depicts the probability of achieving a $CC_i$ fair outcome for each CC under any of the three RSs.

\emph{Take-away.} This analysis carries several important messages. First, it shows that exploration is key in achieving $CC_1$-fair outcomes. Second, although ExtremePA is ex-ante fair it is rarely ex-post fair, thus underlying the importance of looking beyond the number of followers in expectation. Third, our numerical analysis reveals that although exploration always leads to a $CC_1$-fair outcome, the fairness for the other CCs is not guaranteed. However, it suggests that exploratory RSs distribute fairness more homogeneously compared to ExtremePA.

\begin{corollary}
\label{cor:CC1_fair_PA}
Under UR and PA recommendations, $(A^t)_t$ is both ex-post and ex-ante $CC_1$-fair.
\end{corollary}
\begin{proof}
This is an immediate consequence of Theorem \ref{thm:absorbing_PA}. Since $B$ is absorbing iff $b_{., 1} = m$, all absorbing states are $CC_1$ fair (i.e., we have ex-post fairness) and $\mathbb{E}[b_{., 1}] = m \geq \mathbb{E}[b_{., i} ]$ (i.e., ex-ante $CC_1$-fairness).
\hfill$\qed$
\end{proof}


\begin{corollary}
\label{cor:CC1_fair_ExtremePA} 
Under ExtremePA, the probability the outcome is ex-post $CC_1$-fair goes to $1/n$ as $m \to \infty$.
\end{corollary}
\begin{proof}
The final outcome (a) is always $CC_1$-fair if $CC_1$ is the unique most followed CC after the first round, and (b) is $CC_1$-fair only if $CC_1$ is one of the CCs with a maximum number of followers after the first round. So, using the notation from Fig.~\ref{fig:MC_extreme}: 
$p_{\textbf{0}, S_1\cup S_1^*} \leq  \mathbb{P}(CC_1-\text{fair}) \leq p_{\textbf{0}, S_1\cup S_1^*} + p_{\textbf{0}, E}.$
The claim follows, since, as shown previously, $p_{\textbf{0}, S_1\cup S_1^*}\to 1/n$ and $p_{\textbf{0}, E}\to 0$.
\hfill$\qed$
\end{proof}

\begin{corollary}
\label{cor:CC_fair_ExtremePA}
Under ExtremePA, the probability the outcome is ex-post fair goes to $1/n!$ as $m \to \infty$.
\end{corollary}
\begin{proof}
By Lemma~\ref{lem:eq->0}, when $m\to \infty$ the probability of achieving a tie in the number of followers after the first round goes to zero. After ignoring ties, only outcomes with $a^1_{., 1} > \dots > a^1_{., n}$ lead to a fair outcome. By symmetry, all strict orderings of $(a^1_{., i})_i$ have an equal probability.
Since there are $n!$ such orderings, $\mathbb{P}(a^1_{., 1} >  \dots > a^1_{., n}) \to 1/n!$ as $m\to \infty$. The conclusion follows.
\hfill$\qed$
\end{proof}

\begin{figure}%
    \centering
    \includegraphics[width=\textwidth]{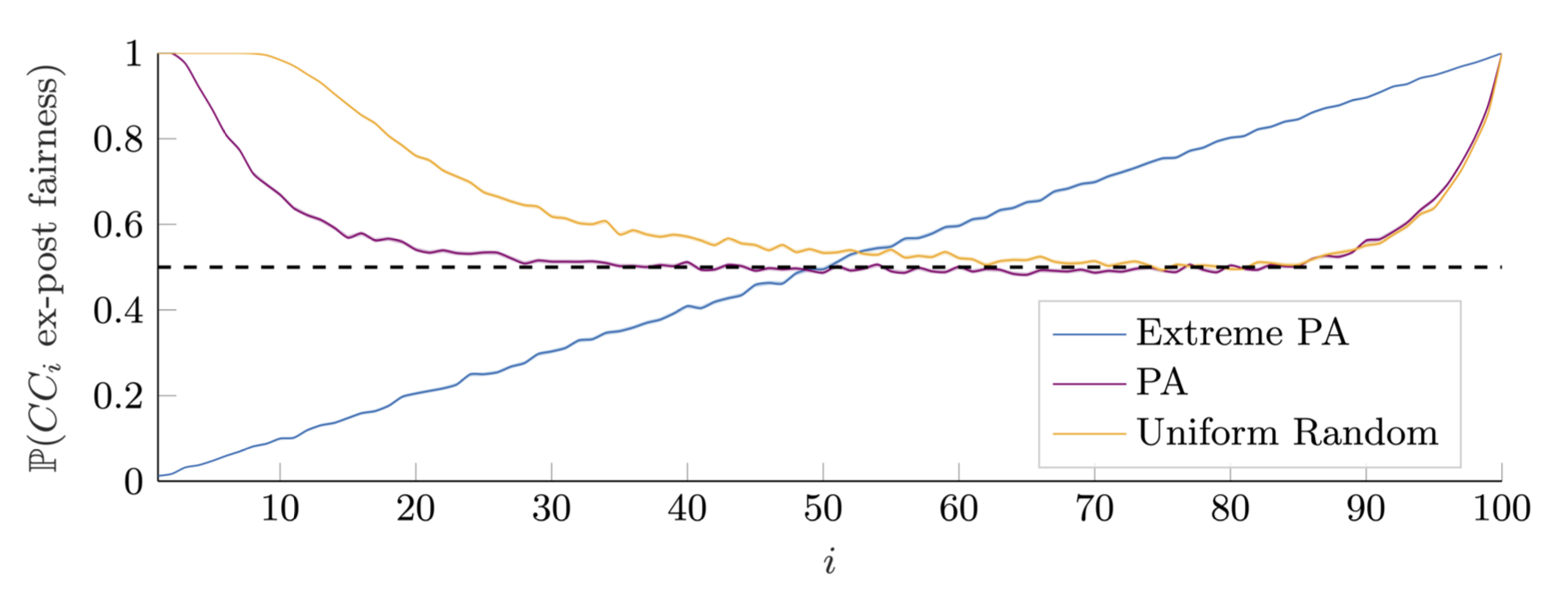}
\caption{
Sampling probability of ex-post $CC_i$-fairness, under the three RSs: ExtremePA (blue), PA (purple), and UR (orange). For each RS, we run $10000$ simulations until convergence with $n=100$ and $m=10000$. The 
dashed line denotes a reference value of $0.5$. Fairness under ExtremePA increases from best- to lowest-quality CCs. PA and UR achieve higher fairness for the best CCs.
}
\label{fig.fairness}
\end{figure}

The behavior depicted in Fig.~\ref{fig.fairness} is consistent with the theoretical results on ex-post fairness. In particular, while UR and PA guarantee $CC_1$-fairness at convergence, ExtremePA only achieves it with a low probability (about $1\%$). More precisely, the probability of ex-post $CC_i$-fairness grows linearly from $1/n$ to $1$ with the quality index $i$. In contrast, UR and PA generally achieve more (less) fair outcomes for the top (lowest) half of CCs. Between the two, the most exploratory recommendation strategy (UR) is more fair for most high-quality CCs, while PA is fair only for a small number of top (and bottom) quality CCs. 

\begin{theorem}
\label{thm:ex_ante_ExtremePA}
Under ExtremePA, the final outcome is ex-ante fair.
\end{theorem}
\begin{proof}
The proof is based on the following observation: when $CC_i$ becomes the most followed $CC$, they will be eventually followed by all CCs who did not follow a better-quality CC before. For example, if $CC_2$ is the most followed after the first round, all users (except those who were recommended $CC_1$) will follow $CC_2$ after round 2. Differently, if $CC_1$ was the most followed, then everybody will end up following $CC_1$ next, while if $CC_3$ was the most followed, then everybody except those who originally followed $CC_1$ or $CC_2$ will follow $CC_3$. This intuitively leads to $CC_2$ having more followers in expectation than $CC_3$ and fewer than $CC_1$. 
For simplicity, we will only formalize this intuition for $n=2$:
\begin{align*}
    \mathbb{E}[a^{\infty}_{., 1}] &= \sum_{k=0}^m \mathbb{P}(a^1_{., 1} = k) \cdot \mathbb{E}[a^{\infty}_{., 1}| a^1_{., 1} = k] 
    = \sum_{k=0}^{\left[\frac{m-1}{2} \right]} \frac{\binom{k}{m}}{n^m} \cdot k + \sum_{k = \left[\frac{m-1}{2} \right] +1}^{m} \frac{\binom{k}{m}}{n^m} \cdot m.
\end{align*}
As the sum is larger than $\mathbb{E}[a^{\infty}_{., 2}]= \sum_{k = 0}^{m} \frac{\binom{k}{m}}{n^m} \cdot k$, ExtremePA is ex-ante fair. \hfill$\qed$

\end{proof}


\section{Conclusion}
In this work we analyzed the role recommendations play in UGC-based social networks with respect to the individual fairness for CCs.
To do so, we (a) extended prior models with a RS-framework, (b) introduced a non-exploratory RS, and (c) defined ex-ante and ex-post measures of fairness. 
Our results showed that the network formation process is an absorbing Markov Chain with different absorbing states, expected times, and fairness guarantees depending on the RS. In particular, the expected absorbing time under ExtremePA is bounded by $2$, i.e., much faster than the ones under UR or PA (which are linearly or sub-linearly increasing in $m$).
Furthermore, while all studied RSs guarantee ex-ante fairness, ex-post fairness is rarely attainable: without exploration (i.e., under ExtremePA), ex-post fairness is achieved with probability $1/n!$, and even with exploration (i.e., under UR or PA), ex-post fairness is guaranteed only for the best (and lowest) CCs. 
In essence, exploration in RSs trades faster absorption times for higher probabilities of achieving fair outcomes for the best CCs.

%
%
\bibliographystyle{unsrt}
\bibliography{main}

\end{document}